\documentclass[letter,11pt]{article}
\usepackage[margin=1in]{geometry}
\usepackage{amsthm}
\usepackage{enumerate}
\usepackage{mathtools}
\usepackage{latexsym}
\usepackage{amsmath,amssymb,amsfonts}
\usepackage{color}
\usepackage{bbold}
\usepackage{xspace}
\usepackage{lettrine}
\usepackage{fancybox}
\usepackage{hyperref}

\newtheorem{theorem}{Theorem}
\newtheorem{lemma}{Lemma}

\newtheorem{observation}{Observation}

\newcounter{alg}

\newcommand{\eps}{\ensuremath{\varepsilon}}

\usepackage{xspace}
\usepackage{tcolorbox}

\newcommand{\rpay}[2]{\ensuremath{\mathcal{R}(#1,#2)}\xspace}
\newcommand{\cpay}[2]{\ensuremath{\mathcal{C}(#1,#2)}\xspace}

\newcommand{\argmax}{\ensuremath{\text{argmax}}\xspace}

\newcommand{\simplex}{\ensuremath{\Delta^n}\xspace}
\newcommand{\supp}{\mathrm{supp}}
\newcommand{\poly}{\mathrm{poly}}

\newcommand{\xbf}{\ensuremath{\mathbf{x}}\xspace}
\newcommand{\ybf}{\ensuremath{\mathbf{y}}\xspace}
\newcommand{\xbfs}{\ensuremath{\mathbf{x}^*}\xspace}
\newcommand{\ybfs}{\ensuremath{\mathbf{y}^*}\xspace}
\newcommand{\xbfh}{\ensuremath{\hat{\mathbf{x}}}\xspace}
\newcommand{\ybfh}{\ensuremath{\hat{\mathbf{y}}}\xspace}
\newcommand{\xbfp}{\ensuremath{\mathbf{x}'}\xspace}
\newcommand{\ybfp}{\ensuremath{\mathbf{y}'}\xspace}

\newcommand{\zbf}{\ensuremath{\mathbf{z}}\xspace}
\newcommand{\wbf}{\ensuremath{\mathbf{w}}\xspace}

\newcount\Comments  
\definecolor{darkgreen}{rgb}{0,0.6,0}
\newcommand{\kibitz}[2]{\ifnum\Comments=1{\color{#1}{#2}}\fi}

\begin{document}
\title{\bf A Polynomial-Time Algorithm for 1/2-Well-Supported Nash Equilibria in Bimatrix Games}

 \author{Argyrios Deligkas$^1$, Michail Fasoulakis$^{2,3}$, Evangelos Markakis$^3$\\
 Royal Holloway, University of London$^1$\\
 Foundation for Research and Technology-Hellas (FORTH)$^2$\\
 Athens University of Economics and Business$^3$}
\date{}

\maketitle
\begin{abstract}
Since the seminal PPAD-completeness result for computing a Nash equilibrium even in two-player games, an important line of research has focused on relaxations 
achievable in polynomial time.
In this paper, we consider the notion of $\eps$-well-supported Nash equilibrium, where $\eps \in [0,1]$ corresponds to the approximation guarantee.
Put simply, in an $\eps$-well-supported equilibrium, every player chooses with positive probability actions that are within $\eps$ of the maximum achievable payoff, against the other player's strategy.
Ever since the initial approximation guarantee of 2/3 for well-supported equilibria, which was established more than a decade ago, the progress on this problem has been extremely slow and incremental. 
Notably, the small improvements to 0.6608, and finally to 0.6528, were achieved by algorithms of growing complexity.
Our main result is a simple and intuitive algorithm, that improves the approximation guarantee to 1/2. 
Our algorithm is based on linear programming and in particular on exploiting suitably defined zero-sum games that arise from the payoff matrices of the two players.
As a byproduct, we show how to achieve the same approximation guarantee in a query-efficient way.
\end{abstract}

\setcounter{page}{1}
\section{Introduction}
The computation of a Nash equilibrium has been one of the most fundamental problems in the intersection of computer science and economics. The non-constructive proof of existence by Nash for finite games \cite{Nas51}, motivated a significant volume of works to focus on the quest for efficient algorithms. 
This quest was terminated by the landmark results of \cite{DGP06,CDT09}, establishing PPAD-completeness, even for 
two-player games. 

The above intractability results naturally led to the study of polynomial-time algorithms for approximate Nash equilibria.
But as this is not a typical optimization problem with a single objective to optimize, there exist in the literature different ways of defining approximate solutions. Among them, the two most prevalent notions of approximability are those of $\eps$-Nash equilibria and $\eps$-well-supported Nash equilibria ($\eps$-WSNE), for $\eps>0$, assuming that all payoffs of the game have been scaled to be in $[0,1]$.
In an $\eps$-Nash equilibrium, the expected payoff of any player is at most $\eps$ less than the best-response payoff. In an $\eps$-well-supported Nash equilibrium, any pure strategy that is played with a positive probability, should have a payoff of at most $\eps$ less than the best-response payoff. It is easy to see that the latter is a stronger approximation concept, as 
any $\eps$-well-supported Nash equilibrium is an $\eps$-Nash equilibrium, but the other direction does not always hold.

The research efforts on these two notions have not evolved equally well. For $\eps$-Nash equilibria, there was a steady progress, starting with a $\frac{3}{4}$-approximation in \cite{KPS06}, and gradual improvements to $\frac{1}{2}$ by \cite{DMP09}, $0.382$ in \cite{DMP07}, $0.364$ by \cite{BBM10}, and $(0.3393+\delta)$ by \cite{TS08}, for any constant $\delta>0$. Finally, a very recent improvement by \cite{DFM22}, has obtained the currently best approximation of $(\frac{1}{3}+\delta)$, for any constant $\delta>0$. On the other hand, the first result for $\eps$-well-supported Nash equilibria was given by Kontogiannis and Spirakis \cite{KS-wsne}, with a $\frac{2}{3}$-approximation. This was obtained by a quite ``clean'' and easy to implement algorithm, requiring only a single call to a linear programming solver. After this, the progress on this front has been very incremental and with algorithms of growing complexity. A refinement of the algorithm by \cite{KS-wsne} was analyzed in \cite{FGSS16-approximate-WSNE}, obtaining a $0.6608$-approximation. Later on, a more involved algorithm was provided by \cite{czumaj2019distributed-wsne}, achieving
an additional improvement to 0.6528. The guarantee of 0.6528 has been the state of the art for the last 7 years~\cite{distibuted-arxiv}.

For both approximation notions, it is well known that there exist quasi-polynomial time algorithms \cite{LMM03,KS-wsne} for any constant $\eps >0$, running in time $n^{O(\log{n}/\epsilon^2)}$, for a game with $n$ available pure strategies per player. Quite interestingly, for well-supported equilibria, there exists an improved quasi-polynomial time algorithm for computing a $(\frac{1}{2}+\delta)$-approximation, for any constant $\delta>0$ \cite{FM19}, making the exponent dependent on $\log{\log{n}}$, rather than $\log{n}$. This can be viewed as supporting evidence that further improvements on the running time might be feasible for obtaining a $(\frac{1}{2}+\delta)$-well-supported approximation, and even a polynomial time algorithm could be anticipated.

Finally, regarding impossibility results, we know by the work of Rubinstein~\cite{Rub16-PETH} that, assuming the exponential-time hypothesis for PPAD, then $\eps$-Nash equilibria, and consequently $\eps$-well-supported equilibria as well, require quasi-polynomial time, when $\eps$ is below some very small, yet unspecified, constant. 
This rules out a PTAS for both approximation notions. But especially for well-supported equilibria, it also leaves open a very large gap from the 0.6528-approximation.

\subsection{Our contribution}
We improve upon the state of the art on the computation of approximate well-supported Nash equilibria in two respects. 
Firstly, we derive a polynomial-time algorithm (Algorithm 1) that, for every constant $\delta$, computes a $(\frac{1}{2}+\delta)$-WSNE. 
Our algorithm is more intuitive and conceptually much simpler compared to the algorithm of~\cite{czumaj2019distributed-wsne} that finds a $0.6528$-WSNE.
Then, we use  Algorithm 1 as a basis for a query-efficient algorithm (Algorithm 2) that achieves the same approximation guarantee.

\paragraph{\bf Algorithm 1 (main result):} 
To provide a high level description, the first step in Algorithm 1 is to compute a Nash equilibrium in the zero-sum games $(R, -R)$ and $(-C,C)$, where $R$ and $C$ are the payoff matrices of the row and column player respectively.
These equilibrium-strategies, and their corresponding payoffs, are utilized to partition the remaining analysis into three cases that the algorithm considers sequentially. The first and easier case occurs when the solutions of the two zero-sum games induce a strategy profile, where the maximum possible payoff for each player under the chosen strategy of their opponent is ``low'', i.e., it is bounded by $\frac{1}{2}$. This directly yields a $1/2$-approximation. The remaining two cases (referred to as the low-high and the high-high cases) are more interesting and in these cases we utilize either a linear system solver, or a support enumeration argument, to construct the desirable strategy profile.
 
\paragraph{\bf Algorithm 2: query-efficient algorithm.} In order to convert Algorithm 1 into a query-efficient protocol we utilize the fact that \eps-WSNEs in zero-sum games can be found with $O(\frac{n\cdot\log n}{\eps^2})$ queries. Then, using the approximate well-supported Nash equilirbia we found, we resort to {\em sampling} in order to create a subgame with $O(\log n)$ rows and $O(n)$ columns (or the other way around, depending on the case we are handling). This subgame is guaranteed to have a strategy profile that is a $(\frac{1}{2}+\delta+3\eps)$-WSNE for the original game. In order to find such a profile the algorithm simply queries the whole subgame and it proceeds as our main algorithm.

\subsection{Further related work}
Apart from results on general two-player games, there have been several studies focusing on the approximability of Nash equilibria in specific classes of games. For instance, for
 constant-rank games, where the matrix defined by the sum of the two payoff matrices has constant rank~\cite{adsul2011rank1,kannan2010-rank-games,Mehta18-constant-rank}; for win-lose games \cite{chen2007-win-lose,codenotti2005-win-lose,liu2018-win-lose}; for sparse games \cite{chen2006-sparse,Barman18-caratheodory-qptas}, for imitation games ~\cite{mclennan2010-imitation,mclennan2010simple-imitation,MurhekarMehta20-imitation}; for random games~\cite{barany2007nash-random,panagopoulou2014random}; for symmetric games \cite{CFJ14,KS11} and for games with symmetric payoff matrices \cite{CFJ17a}. 

There have also been results on new quasi-polynomial algorithms (QPTAS). We mention three new QPTASs that have been obtained:~\cite{Barman18-caratheodory-qptas} gave a refined, parameterized, approximation scheme;~\cite{BabichenkoBP17} gave a QPTAS that can be applied to multi-player games as well;~\cite{DeligkasFMS22-eps-ETR} gave a more general approach for approximation schemes for the existential theory of the reals.
Furthermore, more negative results for \eps-NE were derived:~\cite{KothariM18-sum-of-squares} gave an unconditional lower bound, based on the sum of squares hierarchy;~\cite{BoodaghiansBHR20-Smoothed-2Nash} proved PPAD-hardness in the smoothed analysis setting;~\cite{austrin2011inapproximability,BravermanKW15-QP-LB,DeligkasFS18-constrained-QP-LB} gave quasi-polynomial time lower bounds for constrained \eps-NE, under the exponential time hypothesis.
\section{Preliminaries}
\label{sec:prelims}

In what follows, for any natural number $n$, we denote by $[n]$ the set $\{1, 2, \ldots, n\}$ and by $\simplex$ the $(n-1)$-dimensional simplex.  
An $n \times n$ {\em bimatrix} game $(R,C)$ is defined by two {\em payoff} matrices $R$ and $C$ of size $n \times n$ each: $R$ defines the payoffs of the {\em row} player and $C$ defines the payoffs of the {\em column} player. 
When the row player picks a row $i \in [n]$ and the column player picks a column $j \in [n]$, then they receive a payoff of $R_{ij}$ and $C_{ij}$, respectively. 
We follow the usual assumption in the relevant literature that the matrices are normalized, so that all entries are in $[0, 1]$.

A {\em mixed strategy} is a probability distribution over $[n]$. We use $\xbf \in \simplex$ to denote a mixed strategy for the row player and $\ybf \in \simplex$ to denote a mixed strategy for the column player. 
If \xbf and \ybf are mixed strategies for the row and the column player respectively, then we call $(\xbf,\ybf)$ a {\em strategy profile}. We denote by $e_i$ the $n$-dimensional vector that has 1 at index $i$ and 0 elsewhere. Hence, $e_i$ corresponds to a pure strategy, where a player assigns probability one to play the pure strategy $i$. A useful notion in equilibrium computation is the {\em support} of a strategy \xbf, denoted by $\supp(\xbf)$, which is the set of pure strategies that are played with positive probability under \xbf. Formally, $\supp(\xbf):= \{i \in [n]: \xbf(i) > 0 \}$.

Given a strategy profile $(\xbf,\ybf)$, the expected payoff of the row player is $\rpay{\xbf}{\ybf}:= \xbf^TR\ybf$, and the expected payoff of the column player is $\cpay{\xbf}{\ybf}:=\xbf^TC\ybf$. Hence, for a pure strategy $e_i$, the term $\rpay{e_i}{\ybf}:=\sum_j  R_{ij}y_j$, denotes the expected payoff of the row player, when she plays the pure strategy $i$ against strategy \ybf of the column player. Similarly, $\cpay{\xbf}{e_j}$ is the expected payoff of the column player when she plays the pure strategy $j$ against \xbf. 
A pure strategy is a {\em best-response strategy} for a player against a chosen strategy of her opponent, if it maximizes her expected payoff; a pure strategy is an $\eps$-best response, for $\eps \in [0,1]$, if it achieves a payoff that is at most \eps\, less than the maximum possible payoff against the opponent's chosen strategy.

\medskip
\noindent {\bf Well-supported Nash equilibria.} 
A strategy profile $(\xbf,\ybf)$ is an \eps-{\em well-supported Nash equilibrium}, henceforth \eps-WSNE, for some $\eps \in [0,1]$, if every player plays with positive probability only pure strategies that are \eps-best responses. 
Put formally, $(\xbf, \ybf)$ is an \eps-WSNE if the following conditions are satisfied:
\begin{align*}
    \rpay{e_i}{\ybf} \geq \max_{k \in [n]} \rpay{e_k}{\ybf} - \eps \quad \text{for all}~i \in \supp(\xbf);\\
    \cpay{\xbf}{e_j} \geq \max_{k \in [n]} \cpay{\xbf}{e_k} - \eps \quad \text{for all}~j \in \supp(\ybf).
\end{align*}
If $\eps=0$, then 0-WSNE is an {\em exact} Nash equilibrium.

\medskip
\noindent {\bf $k$-uniform strategies.} 
For any natural number $k$, a {\em $k$-uniform strategy} plays each pure strategy with probability that is a multiple of $\frac{1}{k}$; i.e. $\xbf$ is $k$-uniform if for every $i \in [n]$ it holds that $\xbf(i) \in \{0, \frac{1}{k}, \frac{2}{k}, \ldots, \frac{k-1}{k}, 1\}$. If both $\xbf$ and $\ybf$ are $k$-uniform strategies, then we will say that $(\xbf, \ybf)$ is a $k$-uniform strategy profile. 
We will use $k$-uniform strategies for different values of $k$ in several parts of our algorithms. 
Observe that for any fixed value of $k$, the size of the support of a $k$-uniform strategy is at most $k$. This implies that in a $n \times n$ bimatrix game, there are at most $O(n^k)$ $k$-uniform strategy profiles. In addition, $k$-uniform strategies are used to produce {\em sampled strategies} as explained below.

\medskip
\noindent {\bf Sampled strategies.}
Given a strategy $\xbf$, a $k$-uniform {\em sampled strategy} $\xbf_s$ is produced by taking $k$ independent samples from $[n]$ according to $\xbf$, and by setting $\xbf_s(i)$ equal to the frequency with which pure strategy $i$ was sampled. Observe that $\supp(\xbf_s) \subseteq \supp(\xbf)$. 
Sampled strategies can approximate well any arbitrary strategy if enough samples are taken; roughly, if the sampled strategy has to satisfy $m$ constraints with an $\eps$-additive error, then $O(\frac{\log m}{\poly(\eps)})$ samples suffice. 

\medskip
\noindent
Next, we list some lemmas that involve $k$-uniform sampled strategies that will be used later in the analysis of our algorithms. 

\begin{lemma}[Implied by the proof of Theorem 1 of \cite{LMM03}]
\label{lem:uniform-approximation}
Let $(\xbf, \ybf)$ be a strategy profile of an $n \times n$ bimatrix game. Then, there exist $O(\frac{\log n}{\eps^2})$-uniform sampled strategies $\xbf_s, \ybf_s$ such that:
\begin{align*}
    |\rpay{\xbf}{\ybf} - \rpay{\xbf_s}{\ybf_s}| \leq \eps \quad & \text{and} \quad
    |\rpay{e_i}{\ybf} - \rpay{e_i}{\ybf_s}| \leq \eps \quad \forall i \in [n];\\
    |\cpay{\xbf}{\ybf} - \cpay{\xbf_s}{\ybf_s}| \leq \eps \quad & \text{and} \quad |\cpay{\xbf}{e_j} - \cpay{\xbf_s}{e_j}| \leq \eps \quad \forall j \in [n].
\end{align*}
\end{lemma}

\begin{lemma}
\label{lem:queries-sampled}
Let $(\xbf, \ybf)$ be an $\eps$-WSNE of an $n \times n$ bimatrix game and let $\xbf_s, \ybf_s$ be $O(\frac{\log n}{\eps^2})$-uniform strategies sampled from $\xbf$ and $\ybf$ respectively. Then, $(\xbf_s, \ybf_s)$ is a $3\eps$-WSNE for the game.
\end{lemma}
\begin{proof}
Assume that we are considering the bimatrix game $(R,C)$. We will focus on the row player, since the analysis for the second player is identical. Recall that since $\xbf_s$ is sampled from $\xbf$ it holds that $\supp(\xbf_s) \subseteq \supp(\xbf)$. So, consider an $i \in \supp(\xbf_s)$:
\begin{align*}
    \rpay{e_i}{\ybf_s} & \geq \rpay{e_i}{\ybf} - \eps &\text{(From Lemma~\ref{lem:uniform-approximation})}\\
    & \geq \max_{k \in [n]} \rpay{e_k}{\ybf} - 2\eps & \text{(Since $(\xbf, \ybf)$ is an $\eps$-WSNE)}\\ 
    & \geq \rpay{e_{k^*}}{\ybf} - 2\eps & \text{($k^* \in \argmax_k{\rpay{e_{k}}{\ybf_s}}$)}\\ 
& \geq \rpay{e_{k^*}}{\ybf_s} - 3\eps &\text{(From Lemma~\ref{lem:uniform-approximation}).}\\ 
& = \max_{k \in [n]} \rpay{e_k}{\ybf_s} - 3\eps &
\end{align*}
\end{proof}

\section{The algorithm}
In this section we provide our main algorithm, that computes a $(\frac{1}{2}+\delta)$-WSNE in polynomial time. This algorithm will form the basis for obtaining an additional algorithm that still computes a $(\frac{1}{2}+\delta)$-WSNE and is also query efficient.

\begin{tcolorbox}[title={Algorithm 1: Main Algorithm for $(\frac{1}{2}+\delta)$-WSNE}]
{\bf Input:} A bimatrix game $(R,C)$ and constants $\delta \in (0,1],$ and $\kappa(\delta) = \left \lceil \frac{2\ln{(1/\delta)}}{\delta^2} \right \rceil$.\\
\medskip
{\bf Output:} A $(\frac{1}{2}+\delta)$-WSNE for $(R,C)$.\\
\medskip
{\bf 1.} Compute a Nash equilibrium $(\xbfs,\ybfs)$ for the zero-sum game $(R,-R)$.\\
\medskip
{\bf 2.} Compute a Nash equilibrium $(\xbfh,\ybfh)$ for the zero-sum game $(-C,C)$.\\
{\bf 3.} If $\rpay{\xbfs}{\ybfs} \geq \cpay{\xbfh}{\ybfh}$:
\begin{itemize}
    \item[{\bf (a)}] If $\rpay{\xbfs}{\ybfs} \leq \frac{1}{2}$, then return $(\xbfh, \ybfs)$. 
    \item[{\bf (b)}] Else, if there exists a strategy $\xbfp$ such that $\supp(\xbfp) \subseteq \supp(\xbfs)$, and $\cpay{\xbfp}{e_j} \leq \frac{1}{2}$ for every $j \in [n]$, then return $(\xbfp, \ybfs)$. 
    \item[{\bf(c)}] Else, exhaustively search over all $\kappa(\delta)$-uniform strategy profiles and find a $(\frac{1}{2}+\delta)$-WNSE.
\end{itemize}
{\bf 4.} Else if $\cpay{\xbfh}{\ybfh} > \rpay{\xbfs}{\ybfs}$:
\begin{itemize}
    \item[{\bf (a)}] If $\cpay{\xbfh}{\ybfh} \leq \frac{1}{2}$, then return $(\xbfh, \ybfs)$. 
    \item[{\bf (b)}] Else if there exists a strategy $\ybfp$ such that $\supp(\ybfp) \subseteq \supp(\ybfh)$, and $\rpay{e_i}{\ybfp} \leq \frac{1}{2}$ for every $i \in [n]$, then return $(\xbfh, \ybfp)$. 
    \item[{\bf(c)}] Else, exhaustively search over all $\kappa(\delta)$-uniform strategy profiles and find a $(\frac{1}{2}+\delta)$-WNSE.
\end{itemize}
\end{tcolorbox}

\begin{theorem}
\label{thm:main}
For any constant $\delta > 0$, Algorithm 1 computes in polynomial time a $(\frac{1}{2}+\delta)$-WSNE.
\end{theorem}

Before proving Theorem \ref{thm:main}, we provide first a high level description.
Algorithm 1 partitions the solution space into three cases that considers then sequentially. 
If the algorithm fails to find a solution within a specific case, this rules out the existence of certain strategy profiles and proceeds to the next case, where it exploits the extra constraints that are imposed from the failure of the previous cases. 
Before reaching the case-analysis though, the algorithm computes a Nash equilibrium $(\xbfs, \ybfs)$ for the zero-sum game $(R, -R)$ and a Nash equilibrium $(\xbfh, \ybfh)$ for the zero-sum game $(-C,C)$.
These equilibrium-strategies, and their corresponding payoffs, are utilized to provide bounds for the payoffs in the game $(R,C)$. 
\begin{itemize}
    \item {\bf Case (a): low payoffs.} 
    Firstly, the algorithm checks if the solutions of the two zero-sum games can be used to produce a strategy profile, $(\xbfh, \ybfs)$, where the maximum payoff for every player under the chosen strategy of their opponent is ``low'', i.e., it is bounded by $\frac{1}{2}$. 
    If this is the case, then under the strategy profile $(\xbfh, \ybfs)$, the regret of {\em any} action is bounded by $\frac{1}{2}$, and we have a $\frac{1}{2}$-WSNE.
    \item {\bf Case (b): low-high payoffs.} If the first case fails, then the algorithm deduces that there is a strategy profile among $(\xbfs,\ybfs)$ and $(\xbfh, \ybfh)$, where at least one of the players is guaranteed a ``high'' payoff from every action they play with positive probability, i.e., every action in her support yields a payoff of at least $\frac{1}{2}$.
    Then, it formulates a linear-feasibility system that seeks a strategy for the high-payoff player yielding ``low'' maximum payoff for her opponent.
    For example, if the algorithm finds out that under $(\xbfs, \ybfs)$, the row player is guaranteed to get high payoff, it fixes the strategy $\ybfs$ for the column player, since every action that is assigned positive probability under $\xbfs$ is guaranteed to yield a payoff of at least $\frac{1}{2}$ for the row player. 
    Then, it formulates the linear-feasibility system that checks if there is a strategy $\xbfp$ that restricts the row player to place probability only on actions in the support of $\xbfs$ and at the same time bounds the maximum payoff of the column player by $\frac{1}{2}$. 
    If the linear system has indeed a feasible solution, the algorithm returns the profile $(\xbfp,\ybfs)$, and this will be a $\frac{1}{2}$-WSNE, by the preceding discussion. 
    \item {\bf Case (c): high payoffs.} If the previous two cases fail, then we can prove that there exists a strategy profile $(\xbf, \ybf)$, under which both players are playing only actions that yield a payoff of at least $\frac{1}{2}$, thus $(\xbf, \ybf)$ is a $\frac{1}{2}$-WSNE. 
    In order to prove this, we first define a ``subgame'' of $(R,C)$ that depends on the Nash equilibria we have computed for the zero-sum games $(R,-R)$ and $(-C,C)$. Then, using the properties of Nash equilibria, and the fact that the previous two cases of the algorithm failed, we prove that in {\em any} Nash equilibrium $(\xbf, \ybf)$ of this subgame, both players get a payoff of at least $\frac{1}{2}$.
    On the other hand though, it is not straightforward to find such a strategy profile. Luckily, we can show that we can approximate this by a strategy profile $(\wbf, \zbf)$, where each player plays with positive probability {\em constantly-many} actions and every such action yields a payoff of at least $\frac{1}{2}-\delta$. Thus, we obtain a $(\frac{1}{2} + \delta)$-WSNE by exhaustively searching over the strategy profiles with constant support.
\end{itemize}

\subsection{Proof of Theorem~\ref{thm:main}}
Firstly, observe that Steps 1 and 2 can be performed in polynomial since it is known that zero-sum games can be solved via a linear program. The remaining steps can also be performed in polynomial time, as we will argue in the sequel. 
In what follows we will assume that without loss of generality $\cpay{\xbfh}{\ybfh} \leq \rpay{\xbfs}{\ybfs}$ and thus we will focus on the analysis of Step 3, since the analysis of Step 4 is symmetric.

Before we proceed with the analysis of Steps 3(a) - 3(c), we state the following simple observation for Nash equilibria of zero-sum games.

\begin{observation}
\label{obs:0-sum}
Let $(\xbfs, \ybfs)$ be a Nash equilibrium of the zero-sum game $(R,-R)$. Then, for any other strategy $\ybf$ of the column player, we have:
$$ \rpay{\xbfs}{\ybfs} \leq  \rpay{\xbfs}{\ybf}. $$
Similarly, for an equilibrium $(\xbfh, \ybfh)$ of the zero-sum game $(-C,C)$ we have that for any strategy $\xbf$ of the row player:
$$ \cpay{\xbfh}{\ybfh} \leq  \cpay{\xbf}{\ybfh}. $$
\end{observation}
\begin{proof}
This follows from the fact that when we have an equilibrium $(\xbfs, \ybfs)$ in a zero-sum game, the strategy $\xbfs$ is a max-min strategy, and therefore the row player can guarantee at least the value of the game by playing $\xbfs$, independently of the column player's choice of strategy. The second inequality is proved in the same way by applying the same argument for the column player in the zero-sum game $(-C,C)$.
\end{proof}

Now we are ready to prove the correctness of our algorithm. We will consider every case separately. For the cases 3(a) and 3(b), the analysis is rather straightforward.

\begin{lemma}[Case 3(a)]
\label{lem:3a}
If $\cpay{\xbfh}{\ybfh} \leq \rpay{\xbfs}{\ybfs} \leq \frac{1}{2}$, then $(\xbfh,\ybfs)$ is a $\frac{1}{2}$-WSNE for the game $(R,C)$.
\end{lemma}
\begin{proof}
Since $(\xbfs, \ybfs)$ is a Nash equilibrium of the game $(R,-R)$, we know that for every $i \in [n]$, it holds that $\rpay{e_i}{\ybfs} \leq \rpay{\xbfs}{\ybfs} \leq \frac{1}{2}$. Hence, the best-response payoff of the row player against $\ybfs$ is at most $\frac{1}{2}$, and thus, {\em any} strategy for the row player satisfies the constraints of $\frac{1}{2}$-WSNE, when the column player plays $\ybfs$. 
Similarly, for the column player we get that for every $j \in [n]$, $\cpay{\xbfh}{e_j} \leq \cpay{\xbfh}{\ybfh} \leq \rpay{\xbfs}{\ybfs} \leq \frac{1}{2}$. So, the best-response payoff of the column player against $\xbfh$ is at most $\frac{1}{2}$, and thus, any mixed strategy for the column player satisfies the constraints of $\frac{1}{2}$-WSNE, when the row player selects $\xbfh$.
\end{proof}

\begin{lemma}[Case 3(b)]
\label{lem:3b}
Suppose that $\cpay{\xbfh}{\ybfh} \leq \rpay{\xbfs}{\ybfs}$ and $\rpay{\xbfs}{\ybfs} > \frac{1}{2}$. If there exists a strategy $\xbfp$ with $\supp(\xbfp) \subseteq \supp(\xbfs)$ such that $\cpay{\xbfp}{e_j} \leq \frac{1}{2}$ for any $j \in [n]$, then $(\xbfp,\ybfs)$ is a $\frac{1}{2}$-WSNE for $(R,C)$.
In addition, we can find such a strategy $\xbfp$, or decide that no such strategy exists, in polynomial time by solving a system of linear equations.
\end{lemma}
\begin{proof}
Assume that such a strategy $\xbfp$ exists.
Recall that since $\xbfs$ is a best response against $\ybfs$, any pure strategy in $\supp(\xbfs)$ is also a best response against $\ybfs$. Thus, any $\xbfp$ with $\supp(\xbfp) \subseteq \supp(\xbfs)$, is a best-response strategy against $\ybfs$. Hence, the strategy profile $(\xbfp, \ybfs)$ satisfies the constraints of $\frac{1}{2}$-WSNE for the row player.

Now, for the column player, observe the following. Since the strategy $\xbfp$ guarantees that $\cpay{\xbfp}{e_j} \leq \frac{1}{2}$, for any $j\in [n]$, this means that the best-response payoff for the column player against $\xbfp$ is at most $\frac{1}{2}$. Thus, any mixed strategy for the column player satisfies the constraints of the $\frac{1}{2}$-WSNE, when the row player selects $\xbfp$.

Finally, it only remains to argue that we can decide in polynomial time if such a strategy exists, and if it does, to compute it efficiently. Indeed, this can be formulated as a feasibility system of linear equations. In particular, it corresponds to finding a solution for the linear system defined by: 
\begin{align*}
\cpay{\xbfp}{e_j} \leq \frac{1}{2} & \quad \text{for every $j \in [n]$}; \\ 
\xbfp(i) \geq 0 & \quad \text{for every $i \in \supp(\xbfs)$};\\
\xbfp(i) = 0 & \quad \text{for every $i \notin \supp(\xbfs)$};\\
\sum_i \xbfp(i) = 1.
\end{align*}
\end{proof}

\noindent
Next we focus on Case 3(c), whose correctness seems less intuitive compared to the other two cases. In order to argue about the correctness of the algorithm, we will focus on the subgame $(R_{\xbfs}, C_{\xbfs})$ of $(R,C)$ that has size $|\supp(\xbfs)| \times n$. 
More specifically, in $(R_{\xbfs}, C_{\xbfs})$ the row player will be constrained to use only pure strategies in $\supp(\xbfs)$, while the column player will be allowed to choose from her complete set of pure strategies. Then, we will prove that the bimatrix game $(R_{\xbfs}, C_{\xbfs})$ possesses a Nash equilibrium $(\xbf,\ybf)$ where both players get payoff at least $\frac{1}{2}$.

\begin{lemma}
\label{lem:3c-aux}
If Algorithm 1 reaches Case 3(c), then for every equilibrium profile $(\xbf,\ybf)$ of the bimatrix game $(R_{\xbfs}, C_{\xbfs})$, it holds that $\rpay{\xbf}{\ybf} > 1/2$ and $\cpay{\xbf}{\ybf} > 1/2$.
\end{lemma}
\begin{proof}
Let $(\xbf, \ybf)$ be a Nash equilibrium of the bimatrix game $(R_{\xbfs}, C_{\xbfs})$. 
Firstly, we will prove that $\rpay{\xbf}{\ybf} > 1/2$. Since the algorithm has reached Case 3(c), it must be true that $\rpay{\xbfs}{\ybfs} > 1/2$. Thus, from Observation~\ref{obs:0-sum}, we get that $\rpay{\xbfs}{\ybf} > 1/2$ for every possible $\ybf$. Hence, it must be true that $\rpay{\xbf}{\ybf} > 1/2$ since otherwise the row player could deviate to $\xbfs$, which is a valid strategy for the game $(R_{\xbfs}, C_{\xbfs})$, and increase her payoff; this would contradict the assumption that $(\xbf, \ybf)$ is a Nash equilibrium.

Before we prove that $\cpay{\xbf}{\ybf} > \frac{1}{2}$, we will prove an intermediate claim. Consider the zero-sum game $(-C_{\xbfs}, C_{\xbfs})$ and let $(\tilde{\xbf}, \tilde{\ybf})$ be a Nash equilibrium for this game. We argue that 
\begin{align}
\label{eq:aux-zs}
\cpay{\tilde{\xbf}}{\tilde{\ybf}} > \frac{1}{2}.    
\end{align}
For the sake of contradiction, assume that $\cpay{\tilde{\xbf}}{\tilde{\ybf}} \leq \frac{1}{2}$. Then, since $(\tilde{\xbf}, \tilde{\ybf})$ is a Nash equilibrium of the game $(-C_{\xbfs}, C_{\xbfs})$, we get that
$\cpay{\tilde{\xbf}}{e_j} \leq \frac{1}{2}$ for every $j \in [n]$. But this would mean that Algorithm 1 satisfies the conditions from Case 3(b) and contradicts the fact that the algorithm reached Case 3(c).

Now we are ready to prove that $\cpay{\xbf}{\ybf} > \frac{1}{2}$. To obtain a contradiction, assume that $\cpay{\xbf}{\ybf} \leq \frac{1}{2}$. From Equation \eqref{eq:aux-zs} and Observation~\ref{obs:0-sum}, we get that $\cpay{\xbf}{\tilde{\ybf}} > \frac{1}{2}$. Hence the column player can deviate from $\ybf$ to $\tilde{\ybf}$ and increase her payoff. But this would contradict the assumption that $(\xbf, \ybf)$ is a Nash equilibrium of the game $(R_{\xbfs}, C_{\xbfs})$.
\end{proof}

\noindent
In addition to Lemma~\ref{lem:3c-aux}, we will need the following lemma that uses $k$-uniform sampled strategies for {\em constant} values of $k$, which was proved in~\cite{CFJ14}.

\begin{lemma}[Theorem 2 from \cite{CFJ14}]
\label{thm:cfj14}
Let $(\xbf,\ybf)$ be a Nash equilibrium for a bimatrix game $(R,C)$. Then, for any constant $\delta > 0$, there exists a $\kappa(\delta)$-uniform strategy profile $(\xbf_s, \ybf_s)$, with  $\kappa(\delta) = \left \lceil \frac{2\ln{(1/\delta)}}{\delta^2} \right \rceil$, where $\xbf_s$ is sampled from \xbf, and $\ybf_s$ is sampled from \ybf such that:
\begin{enumerate}
    \item $\rpay{e_i}{\ybf_s} \geq \rpay{\xbf}{\ybf}-\delta$, for every $i \in \supp(\xbf_s)$;
    \label{enum:1}
    \item $\cpay{\xbf_s}{e_j} \geq \cpay{\xbf}{\ybf}-\delta$, for every $j \in \supp(\ybf_s)$;
    \label{enum:2}
\end{enumerate}
\end{lemma}

Lemma~\ref{thm:cfj14} tells us two things that are useful for our purposes. Firstly, Points~\ref{enum:1}-\ref{enum:2} provide an approximation guarantee for the strategy profile $(\xbf_s,\ybf_s)$ as a well-supported NE. Then, the fact that both $\xbf_s$ and $\ybf_s$ are $\kappa(\delta)$-uniform tells us that we can find this profile by enumeration in polynomial time. The lemma below summarizes the above and it completes the proof of correctness of our algorithm.

\begin{lemma}[Case 3(c)]
\label{lem:3c}
If Algorithm 1 reaches Case 3(c), then we can compute in polynomial time a $\kappa(\delta)$-uniform strategy profile $(\wbf,\zbf)$, with $\kappa(\delta) = \left \lceil \frac{2\ln{(1/\delta)}}{\delta^2} \right \rceil$, that is a $(\frac{1}{2}+\delta)$-WSNE for the game $(R,C)$. 
\end{lemma}
\begin{proof}
The proof of the lemma follows from the combination of Lemma~\ref{lem:3c-aux} and Lemma~\ref{thm:cfj14}. 
From Lemma~\ref{lem:3c-aux}, we get that there exists a strategy profile $(\xbf, \ybf)$, that is a Nash equilibrium for the subgame $(R_{\xbfs},C_{\xbfs})$, and additionally under this strategy profile both players get payoff at least $\frac{1}{2}$.
Then, Lemma~\ref{thm:cfj14}, applied to the profile $(\xbf, \ybf)$ for the game $(R_{\xbfs},C_{\xbfs})$, implies the existence of a strategy profile $(\wbf,\zbf)$, such that:
\begin{itemize}
    \item $\rpay{e_i}{\zbf} \geq \frac{1}{2}-\delta$ for every $i \in \supp(\wbf)$ and $\cpay{\wbf}{e_j} \geq \frac{1}{2}-\delta$ for every $j \in \supp(\zbf)$;
\item $\supp(\wbf) \subseteq \supp(\xbfs)$.
\end{itemize}
This profile satisfies all the constraints we want. 
Observe that the first bullet above implies that $(\wbf,\zbf)$ is a $(\frac{1}{2}+\delta)$-WSNE for the game $(R, C)$; this is straightforward to see since the maximum payoff is at most 1. 
Thus, it only remains to prove that we can compute this strategy profile in polynomial time. This is again easy to see: there are $O(n^{\kappa(\delta)})= n^{O\left( \frac{\ln{(1/\delta)}}{\delta^2} \right)}$ $\kappa(\delta)$-uniform strategy profiles, which is polynomially bounded, for every constant $\delta > 0$, and we can check if such a profile is a $(\frac{1}{2}+\delta)$-WSNE in polynomial time.
\end{proof}

\section{A Query Efficient Algorithm}
\label{sec:query}

In this section we consider the query complexity of the problem and we show how Algorithm 1 can be converted to a query-efficient algorithm that achieves the same approximation for well-supported Nash equilibria using $O(n\cdot \log n)$ payoff queries.

\paragraph{\bf Query Complexity.} In the setting, we assume that the algorithm only knows that the players will play an $n \times n$ game $(R,C)$, but it does not know any of the payoff entries of $R$ and $C$. The algorithm learns the payoff entries via {\em payoff queries}: if the algorithm queries the pure strategy profile $(e_i,e_j)$, it is told the payoffs $R_{ij}$ and $C_{ij}$ the players get under this strategy profile. After each payoff query, the algorithm can perform arbitrary computations --our algorithm will perform polynomial-time computations-- in order to decide the next strategy profile to query (if any). In the end, the algorithm outputs a mixed strategy profile $(\xbf,\ybf)$. The goal is to ensure that $(\xbf, \ybf)$ has a good approximation guarantee, while keeping the number of queries as small as possible.

\medskip
\noindent
In order to prove our main result for this section, we will use the following lemma, proven by Fearnley and Savani~\cite{FS2016-queries}, which essentially shows that with high probability we can find an $\eps$-WSNE of a zero-sum game with $O(\frac{n\cdot \log n}{\eps^4})$ payoff queries.

\begin{lemma}
\label{lem:queries-zero-sum}[Corollary 5.5 from~\cite{FS2016-queries}]
Given an $n \times n$ zero-sum game,
with probability at least $(1-n^{-\frac{1}{8}})(1-\frac{2}{n})^2$, we can compute an $\eps$-WSNE $(\xbf,\ybf)$ using $O(\frac{n\cdot \log n}{\eps^4})$ payoff queries.
\end{lemma}

While Lemma~\ref{lem:queries-zero-sum} suffices to simulate Steps 1 and 2, and Case 3(a) of Algorithm 1, there is no obvious way to implement the remaining cases in a query-efficient way. 
One idea would be to perform an exhaustive search under the constraint that we consider only the subgame $(R_{\xbfs}, C_{\xbfs})$, where $\xbfs$ is the strategy computed at Step 1.  
This would not break the correctness of the algorithm, since from Lemmas~\ref{lem:3b} and~\ref{lem:3c-aux} we know that it suffices to focus on this subgame if the algorithm reaches Cases 3(b) and 3(c). 
This idea would indeed allow us to simulate Case 3(b), and Case 3(c), with $O(n\cdot \log n)$ payoff queries, if $(R_{\xbfs}, C_{\xbfs})$ had $O(n\cdot \log n)$ payoff entries; for example, everything would work fine if $|\supp(\xbfs)| = O(\log n)$. 
However, Lemma~\ref{lem:queries-zero-sum} does not provide any bound on the support size of $\xbf$; in fact it may even return a fully-mixed strategy for every player. Luckily though, we can sample an $O(\log n)$-uniform strategy $\xbfs_s$ that approximates well enough $\xbfs$. 

We describe formally below all the adjustments that need to be done to Algorithm 1. 

\begin{tcolorbox}[title={Algorithm 2: Query-Efficient implementation of Algorithm 1}]
{\bf Input:} A bimatrix game $(R,C)$ and constants $\eps, \delta \in (0,1],$ and $\kappa(\delta) = \left \lceil \frac{2\ln{(1/\delta)}}{\delta^2} \right \rceil$.\\
\medskip
{\bf Output:} A $(\frac{1}{2}+3\eps+\delta)$-WSNE for $(R,C)$ using $O(\frac{n\cdot \log n}{ \eps^4})$ payoff queries.\\
\medskip
{\bf 1.} Use Lemma~\ref{lem:queries-zero-sum} to compute an \eps-WSNE $(\xbfs,\ybfs)$ for the zero-sum game $(R,-R)$.\\
\medskip
{\bf 2.} Use Lemma~\ref{lem:queries-zero-sum} to compute an \eps-WSNE $(\xbfh,\ybfh)$ for the zero-sum game $(-C,C)$.\\
{\bf 3.} If $\rpay{\xbfs}{\ybfs} \geq \cpay{\xbfh}{\ybfh}$, use Lemma~\ref{lem:uniform-approximation} to sample the strategy profile $(\xbfs_s, \ybfs_s)$.\\
 \hspace*{3mm} Query all strategy profiles $(e_i,e_j)$, with $i \in \supp(\xbfs_s)$, and create the subgame $(R_{\xbfs_s}, C_{\xbfs_s})$.
\begin{itemize}
    \item[{\bf (a)}] If $\rpay{\xbfs}{\ybfs} \leq \frac{1}{2}$, then return $(\xbfh, \ybfs)$. 
    \item[{\bf (b)}] Else if there exists a strategy $\xbfp$ such that $\supp(\xbfp) \subseteq \supp(\xbfs_s)$ and $\cpay{\xbfp}{e_j} \leq \frac{1}{2}$ for every $j \in [n]$, then return $(\xbfp, \ybfs_s)$. 
    \item[{\bf(c)}] Else, exhaustively search over $\kappa(\delta)$-uniform strategy profiles of $(R_{\xbfs_s}, C_{\xbfs_s})$ and find a $(\frac{1}{2}+3\eps+\delta)$-WSNE of $(R,C)$. 
\end{itemize}
{\bf 4.} Else if $\cpay{\xbfh}{\ybfh} > \rpay{\xbfs}{\ybfs}$, use Lemma~\ref{lem:uniform-approximation} to sample the strategy profile $(\xbfh_s, \ybfh_s)$.\\
 \hspace*{3mm} Query all strategy profiles $(e_i,e_j)$, with $j \in \supp(\ybfh_s)$, and create the subgame $(R_{\xbfh_s}, C_{\xbfh_s})$.
\begin{itemize}
    \item[{\bf (a)}] If $\cpay{\xbfh}{\ybfh} \leq \frac{1}{2}$, then return $(\xbfh, \ybfs)$. 
    \item[{\bf (b)}] Else if there exists a strategy $\ybfp$ such that $\supp(\ybfp) \subseteq \supp(\ybfh_s)$ and $\rpay{e_i}{\ybfp} \leq \frac{1}{2}$ for every $i \in [n]$, then return $(\xbfh_s, \ybfp)$. 
   \item[{\bf(c)}] Else, exhaustively search over $\kappa(\delta)$-uniform strategy profiles of $(R_{\xbfh_s}, C_{\xbfh_s})$ and find a $(\frac{1}{2}+3\eps+\delta)$-WSNE for $(R,C)$.
\end{itemize}
\end{tcolorbox}

\begin{theorem}
\label{thm:query-alg}
For any constants, $\eps, \delta > 0$, Algorithm 2 computes in polynomial time, and with probability at least $(1-n^{-\frac{1}{8}})(1-\frac{2}{n})^2$, a $(\frac{1}{2}+3\eps+\delta)$-WSNE  using $O(\frac{n\cdot \log n}{\eps^4})$ payoff queries.
\end{theorem}

Note that in order to compute a $(1/2+\delta')$-WSNE for some given $\delta'$, it suffices to set the constants $\eps, \delta$ of Algorithm 2, so that $\delta' = 3\eps + \delta$, e.g., using $\eps = \delta = \delta'/4$ suffices.

Before we prove Theorem \ref{thm:query-alg}, we state first the following simple observation, which follows directly from the definition of 
$\eps$-WSNE.

\begin{observation}
\label{obs:aux}
Fix $\eps\geq 0$, and let $(\xbfs, \ybfs)$ be an \eps-WSNE of the zero-sum game $(R,-R)$. Then, we have: 
\begin{align}
    \label{eq:aux1}
    \rpay{\xbfs}{\ybfs} \geq  \rpay{e_i}{\ybfs} - \eps,~\forall i \in [n], \quad & \text{and} \quad \rpay{\xbfs}{\ybfs} - \eps \leq  \rpay{\xbfs}{e_j},~ \forall j \in [n].
\end{align}
Similarly, for any \eps-WSNE $(\xbfh, \ybfh)$ of the zero-sum game $(-C,C)$ we have:
\begin{align}
    \label{eq:aux2}
    \cpay{\hat{\xbf}}{\hat{\ybf}} - \eps \leq \cpay{e_i}{\hat{\ybf}},~ \forall i \in [n], \quad & \text{and} \quad \cpay{\hat{\xbf}}{\hat{\ybf}} \geq \cpay{\hat{\xbf}}{e_j} - \eps,~ \forall j \in [n].
\end{align}
\end{observation}

\begin{proof}[Proof of Theorem~\ref{thm:query-alg}]
We begin by establishing the number of payoff queries the algorithm requires. 
Steps 1 and 2, according to Lemma~\ref{lem:queries-zero-sum}, require $O(\frac{n \cdot \log n}{\eps^4})$ queries each. 
In addition, Steps 3 and 4 require $O(\frac{n \cdot \log n}{\eps^2})$ queries each; this is because we know from Lemma~\ref{lem:uniform-approximation} that $|\supp(\xbfs_s)| = O(\frac{\log n}{\eps^2})$ and $|\supp(\xbfs_s)| = O(\frac{\log n}{\eps^2})$.

In order to prove the correctness of Algorithm 2, we will follow a similar approach as in the proof of Theorem~\ref{thm:main}. 
Without loss of generality assume that $\rpay{\xbfs}{\ybfs} \geq \cpay{\xbfh}{\ybfh}$, hence the algorithm will proceed according to Step 3. We will distinguish between the three possible cases.
\begin{itemize}
    \item {\bf Case 3(a) occurs.} Then, using exactly the same arguments as in Lemma~\ref{lem:3a}, and by Observation~\ref{obs:aux}, we get that $(\xbfh,\ybfs)$ is a $(\frac{1}{2}+\eps)$-WSNE, since $(\xbfs,\ybfs)$ and $(\xbfh,\ybfh)$ are $\eps$-WSNEs and not exact Nash equilibria.
    \item {\bf Case 3(b) occurs.} Then, we follow the proof of Lemma~\ref{lem:3b} and we get a $\frac{1}{2}$-WSNE. This is because we are using the sampled strategy profile $(\xbfs_s, \ybfs_s)$ which, according to Lemma~\ref{lem:uniform-approximation}, is a $3\eps$-WSNE. Thus, we get that \xbfp will satisfy the constraints of a $3\eps$-WSNE for the row player, while the linear system of Lemma~\ref{lem:3b} will satisfy the constraints of a $\frac{1}{2}$-WSNE.
    \item {\bf Case 3(c) occurs.} Here, we can reprove Lemma~\ref{lem:3c-aux}, using the subgame $(R_{\xbfs_s}, C_{\xbfs_s})$ this time, and get that this subgame possesses a Nash equilibrium $(\xbf, \ybf)$ such that $\rpay{\xbf}{\ybf} > \frac{1}{2}-3\eps$, and $\cpay{\xbf}{\ybf} > \frac{1}{2}-3\eps$; the loss of $3\eps$ comes from the use of the sampled strategies and Lemma~\ref{lem:uniform-approximation}. 
    Then, using verbatim the analysis of Lemma~\ref{lem:3c}, we get that there exists a $\kappa(\delta)$-uniform strategy profile in $(R_{\xbfs_s}, C_{\xbfs_s})$ that is a $(\frac{1}{2}+3\eps+\delta)$-WSNE for $(R,C)$.
\end{itemize}
\end{proof}

\section{Discussion}
\label{sec:discussion}

We have presented a new algorithm for computing a $\frac{1}{2}$-well-supported Nash equilibrium in bimatrix games. 
Our algorithm not only significantly improves the previously-best approximation guarantee, but it is conceptually much simpler and intuitive. As a byproduct, we have showed how we can convert our main algorithm into a query-efficient protocol that achieves the same approximation guarantee with $O(n\cdot \log n )$ queries.

Below we identify three orthogonal directions for future work.
\begin{itemize}
    \item Close the gap between polynomial-time upper bounds and the quasi-polynomial lower bound from~\cite{Rub16-PETH}. While we conjecture that 1/2 is not the limit of polynomial-time tractability for \eps-WSNE, we believe that new techniques and ideas will be required for achieving a better approximation. Indeed, the shifting-probabilities ideas that were used in~\cite{FGSS16-approximate-WSNE} and~\cite{czumaj2019distributed-wsne} had as a natural limit $\eps = 1/2$. Hence, some fresh ideas seem to be required to go below $\frac{1}{2}$.
    \item Prove a matching query-complexity lower bound. In~\cite{FS2016-queries} it was shown that any deterministic algorithm requires $\Omega(n^2)$ payoff queries to find an \eps-NE, and thus an \eps-WSNE, for every $\eps < 1/2$. Steps 1 and 2 from Algorithm 2 are randomized, thus the above-mentioned lower bound is not binding; in~\cite{FS2016-queries} it has been shown that for $\eps = 1/6n$, any randomized algorithm requires $\Omega(n^2)$ queries. We conjecture that any randomized algorithm that finds an \eps-WSNE, with $\eps < 1/2$, requires $\Omega(n^2)$ queries.
    \item Improve the approximation guarantee of $\eps$-WSNE with polylogarithmic communication. In the communication complexity model, introduced in~\cite{goldberg-pastink2014communication}, every player knows her own payoff matrix, but she does not know the entries in the payoff matrix of her opponent. The players have to follow a protocol that works in rounds and in every round the players exchange a bit of information. 
    Algorithm 1 can be converted to a communication protocol that finds a $(\frac{1}{2}+\delta)$-WSNE with $O(n\cdot \log n)$ communication, using a similar approach as in~\cite{czumaj2019distributed-wsne}.
    However, there is a fundamental barrier that does not allow us to implement it with polylogarithmic communication, which is the desirable bound in communication complexity.
    Although we can implement almost every step of the algorithm with $O(\log^2n)$ communication, the last step  seems to inevitably require $O(n)$ communication. In particular, the Case 3(c) of the algorithm can be easily reduced to the so-called \textsc{Disjointness} problem, which is known to be ``hard'' with respect to communication~\cite{roughgarden2016communication}. 
    \end{itemize}

\bibliographystyle{plain}
\bibliography{agt}

\end{document}